\documentclass[times, 11pt,latex8]{article}

\usepackage{amssymb}
\usepackage{amsfonts}
\usepackage{amsmath,bbm}
\usepackage{latexsym}
\usepackage{epsfig, epstopdf, wrapfig}
\usepackage{graphicx, xcolor}
\usepackage{caption}
\usepackage{enumerate}
\usepackage{hyperref}

\newtheorem{propn}{Proposition}{}{}
\newtheorem{thm}{Theorem}{}{}
\newtheorem{lem}{Lemma}{}{}
{}{}
{}{}
\newtheorem{Def}{Definition}{}{}
{}{}
{}{}
\newenvironment{proof}{{\bf Proof:}}{~$\Box$ \\}

\newcommand{\commentout}[1]{}

\newcommand{\Real}{\mathbb{ R}}
\newcommand{\Complex}{\mathbb{ C}}
\newcommand{\modulus}[1]{|\!#1\!|}
\newcommand{\norm}[1]{|\!| #1 |\!|}

\newcommand{\ket}[1]{|#1\rangle}
\newcommand{\bra}[1]{\langle #1|}
\newcommand{\inp}[2]{\langle #1|#2\rangle}
\newcommand{\inpr}[3]{\langle #1|#2|#3\rangle}

\newcommand{\conj}[1]{\overline{#1}}
\newcommand{\hconj}[1]{{#1}^{\dagger}}

\newcommand{\tr}{\text{\tt Tr}}
\newcommand {\pj}[1]{\ket{#1}\!\bra{#1}}
\newcommand {\pjx}[2] {\ket{#1}\bra{#2}}

\newcommand{\acti}{\!\cdot\!}

\newcommand{\unitary}[1]{{\mathcal{U}_#1}}
\def\dmn#1#2{{#1^{(#2)}}}

\newcommand{\be}{\begin{enumerate}}
\newcommand{\ee}{\end{enumerate}}
\newcommand{\beq}{\begin{equation}}
\newcommand{\eeq}{\end{equation}}
\newcommand{\beqx}{\begin{displaymath}}
\newcommand{\eeqx}{\end{displaymath}}
\newcommand{\beqa}{\begin{eqnarray}}
\newcommand{\eeqa}{\end{eqnarray}}
\newcommand{\beqax}{\begin{eqnarray*}}
\newcommand{\eeqax}{\end{eqnarray*}}

\title{Metrics and pseudometrics on unitary groups with applications in quantum information processing}

\author{
Manas ~K. ~Patra\\ Department of Data Science and Analytics\\Central University of Rajasthan, NH-8\\Dist-Ajmer-305817, India }
\date{}
\begin{document}
\maketitle

\begin{abstract}
Metrics and pseudometrics are defined on the group of unitary operators in a Hilbert space. Several explicit formulas are derived. A special feature of the work is investigation of pseudometrics in unitary groups.  The rich classes of pseudometrics have many interesting applications. Three such applications, distinguishibility of unitary operators, quantum coding and quantum search problems are discussed.  
\end{abstract}
\maketitle
\section{Introduction}
In this  work I investigate metrics and pseuudometrics on on groups of unitary operators. A preliminary version of this work  appeared a while ago\cite{Patra06}. It had several typographic errors and lacked details. For one reason or other I could not get back to polish it. The present work is a thorough revision and significant extension of that work. Several new results are added. A notable departure from the earlier work is the emphasis and exploration of rich and varied classes of {\em pseudometrics}. While metrics on various mathematical structures have been extensively studied pseudometrics have received scant attention. Pseudometrics, unfettered by restrictions of strict positivity offer a much richer variety. \footnote{Recall taht a metric on a set $X$ is a real-valued function $d$ on $X\times X$ satisfying $d(x, y)=d(y,x)$ (symmetry), $d(x, y) \leq d(x, z)+d(z, y)$(triangle inequality) and $d(x, y) \geq 0$ and equals 0 if and only if $x=y$ (strict positivity). For a pseudometric strict positivity is not required.}

The usual approach to metrics  on the space of operators on a Hilbert space is to start with a norm,  $A  \rightarrow \norm{A}$ on the operators and then define a metric $d(A, B)= c \norm{A-B}, \; c > 0$. 
This approach is particularly useful if we restrict to affine subspaces of
the space of operators. This is because the norm on the ambient space
induces a norm on the space of operators and the metric is defined in
terms of the latter. However, if we restrict to some {\em subset} of
operators which may not constitute a subspace, the norm-induced metric
may not seem very natural. For two operators $A$ and $B$ the
difference $A-B$ whose norm defines the distance between $A$ and $B$
may take us outside the subset. But the concept of a metric does not
depend upon algebraic operations. In particular, if the relevant subset
is a group we are often interested in {\em invariant} metrics. That is
metrics that remain invariant under left (right) translations by the
group operations. Of course, invariant metrics are known to exist for any compact group \cite{Price}. The (pseudo)metrics introduced in the paper are induced by some norms.  Since  a pseudometric $\delta(U,V)$ maybe 0, these null sets are sometimes interesting for quantum information theory. 

In Section II, I introduce pseudometrics and metrics  on  groups of unitary operators. The pseudometrics are induced by norms on the convex set of density matrices. Hence their properties depend on those of the norm.  A metric $d$ on the projective unitary groups $PU_n$ \footnote{Informally, we obtain $PU_n$ as a quotient of unitary group $\mathcal{U}_n$ identifying unitary operators $U$ and $cU$, $c| = 1$.}\cite{Grove} is obtained by taking {\em supremum} over all states for the trace norm. Many algebraic and geometric properties of the metric are proved. An explicit formula for the metric is derived. Using this formula it is easy to derive an explicit formula between two unitary operators that can be written as tensor products. I then demonstrate the equivalence of $d$ to other metrics obtained by using different norms. Finally, several interesting classes of pseudometrics are discussed. In particular, pseudometrics induced by the standard norms on the set of separable states turns out to be metric on $PU_n$. 

 Section III deals with applications. First, it is is shown that two unitary operators are distinguishable if and only if the distance between them is 1: $d(U, V) =1$. Next I use certain pseudometrics to chatracterize stabilizer subspaces of an abelian group. These subspaces are the {\em null} subspaces of the metric. The third application to quantum search algorithm shows that search problem can be formulated and solved as a problem of  approximation in the metric. 

I conclude the paper with some remarks on other avenues to explore  using the infinite varieties of pseudometrics. 

\section{Metrics and pseudometrics on the unitary group}
First, let us fix some notation. In the following, $\mathcal{H}$ will
denote a complex Hilbert space with a fixed inner product $<,>$. The {\em rays} in $\mathcal{H}_n$ correspond to pure states of a quantum system of dimension $n$. The relevant sets will be sometimes subscripted with the dimension $n$ of the underlying Hilbert space. Thus $\mathcal{U}_n$ denotes the group of
unitary operators in $B(\mathcal{H}_n)$, where the latter denotes the
algebra of linear operators on $\mathcal{H}_n$. The corresponding subset
of hermitian operators will be denoted by
$L(\mathcal{H}_n)$. The special unitary group $S\mathcal{U}_n\subset
\mathcal{U}_n$ is the subgroup of operators with determinant 1. I use the
standard notation \(\Complex \text{ and } \Real\) for the field
of real and complex numbers with usual topology. In $\Complex^n$ the
standard inner product is used. Thus, if \(\alpha = (x_1, \ldots, x_n)^T
\text { and } \beta = (y_1, \ldots, y_n)^T \in \Complex^n\), where
$A^T$ denotes the transpose of the matrix $A$, then
\[ \inp{\alpha}{\beta} \equiv \sum_i \conj{x}_iy_i \]
When the dimension $n$ is fixed or can be inferred we drop the subscript. Finally, $I$ (or $I_n$ if the dimension needs to be specified) will denote the unit operator on $\mathcal{H}$. The Hilbert space norm induces a norm $A\rightarrow
 \norm{A}$, on the space of operators on $\mathcal{H}$, defined by,
\[ \norm{A} = \max_{\norm{\psi}=1} \norm {A\psi} 
 \]
 The induced norm is called the operator norm. 
 If $A$ is normal ($A \text{ and } A^\dagger$ commute) then \(\norm{A}= \max\{|\lambda| \:| \:\lambda \text{ an
 eigenvalue of $A$} \}\). Properties of the operator norm
  may be found in \cite{Bhatia}. The norm induces a metric $d_O(A,B) = c\norm{A-B}$ on
 the space of operators where $c > 0$  is some convenient normalization constant.  We observe that if $A,B$ are unitary operators this metric is not even {\em projective} invariant, that is, $d_O(A,B) \neq d_O(e^{ix}A, B)$ for real $x$. Let us consider the distance between two unitary operators from an operational perspective. Suppose we want to test how close are two unitary operators, $U, V$. In a quantum system we can only obtain probability distributions corresponding to various states and experimental arrangements. A naive `experiment' would be to apply say $V$ first to a (pure) state $\ket{\psi}$ followed by the application of $\hconj{U}$ and then make a projective measurement $\{P_\psi, P_{\psi^\perp}\}$ where $P_\psi$ is the orthogonal projection onto $\ket{\psi}$ and $ P_{\psi^\perp} = I - P_{\psi}$, the complementary projection. So if the state $\ket{\psi}$ is fixed then $|\inpr{\psi}{\hconj{U}V}{\psi}|^2 $  is the transition probability and we can use a probability metric to measure the distance between the distribution $(1,0)$ and $(|\inpr{\psi}{\hconj{U}V}{\psi}|^2 , 1-|\inpr{\psi}{\hconj{U}V}{\psi}|^2)$. We will use the total variational distance between two probability measures which has the  simple expression $d_\psi^2(U,V) \equiv 1- |\inpr{\psi}{\hconj{U}V}{\psi}|^2$ in our case. $d_\psi$ is actually a pseudometric\footnote{ A pseudometric on a set $S$ is function $d: S \times S \rightarrow \Real_+$ that satisfies i. $d(x, x) = 0, d(x, y) = d(y, x) \text{ and } d(x, y) \leq d(x, z) + d(z, y)$.}. We will see this in the next subsection when we generalize to mixed quantum states. 
 
\subsection{General construction of metrics}
Let $\mathcal{D}_n$ be the set of positive semidefinite operators in $L(\mathcal{H}_n)$ with trace 1. $\mathcal{D}_n$ is the set of mixed states. The action of a unitary operator $U$ on $\mathcal{D}_n$ is given by
\beq
U\acti \rho = U\rho U^\dagger 
\eeq
Let $E,F$ be two general quantum operations, that is, $E,F$ are completely positive maps on $\mathcal{D}_n$ or {\em superoperators}\cite{Nielsen}. In the spirit of defining an {\em operational} distance between $E$ and $F$ we would like to characterize it by the effect the operators have on the quantum state as in the toy example above. Thus define 
\beq\label{eq:basic_metric}
\Delta_\rho(E, F)  = \Delta(E\acti\rho, F\acti\rho)
\eeq
where $\Delta$ is metric on $\mathcal{D}_n$. 
Our starting point is the following simple lemma. 
\begin{lem}\label{lem:pseudo1}
The function  $\Delta_\rho$ is a pseudometric on the space of quantum operations. Thus 
\be
\item
$\Delta_\rho(E, F) \geq 0 $
\item
$\Delta_\rho(E, F) = \Delta_\rho(F,E)$
\item
$\Delta_\rho(E, F) \leq \Delta_\rho(E, G) + \Delta_\rho(G, F)$
\ee
Further, if $\Delta$ is invariant with respect to the unitary action---$\Delta(U\acti\rho, \: U\acti\sigma) = \Delta(\rho, \sigma)$---then $\Delta_\rho(UE, UF) = \Delta_\rho(X, Y)$ (invariant with respect to left multiplication). In particular, if $X, Y$ are unitary operators: $\Delta_\rho(X, Y) = \Delta_\rho (I, \hconj{X}Y)$.  
\end{lem}
\begin{proof}
All assertions easily follow from the defining properties of the metric $d$ and its invariance properties. 
\end{proof}
We will assume henceforth that the metric $\Delta$ is unitary invariant. Further, unless explicitly stated otherwise the quantum operations will be assumed to be unitary. Clearly,  $\Delta_\rho(X,Y) = 0$ if and only if $\rho$ satisfies $\hconj{X}Y\acti\rho = \rho$. Hence, if the action of the unitary group is linear (on the space of operators) which we also assume then $\rho$ is an eigenstate of $\hconj{X}Y$ with eigenvalue 1 or alternatively, it is a {\em fixed point}. Any unitary invariant metric on the set of density operators $\mathcal{D}$ will induce a corresponding pseudometric on $\mathcal{U}$ with the above properties. Specifically, we use {\em trace} norm: $\norm{A}_1 = \tr (|A|)$ (the trace of $|A|$) where $|A| = \sqrt{\hconj{A}A}$. 
\def\dmn#1#2{{#1^{(#2)}}}
The trace-norm induces a unitary invariant metric 
\[\dmn{d}{1}(\rho, \sigma) = \frac{\norm{\rho - \sigma}_1}{2}\]
which has close analogy with classical probability distance \cite{Nielsen}. We  use $d_\rho$ to denote the pseudometric induced by the metric $\dmn{d}{1}$ in the state $\rho$: 
\beq\label{eq:tr_dist}
d_\rho(E, F) = \dmn{d}{1}(E\!\cdot \!\rho, F\!\cdot \!\rho)
\eeq
for superoperators $E,F$. In particular, for  unitary operators $U,V$ invariance of $\dmn{d}{1}$ implies $d_\rho(U,V) = d_\rho(I, U^\dagger V)$. If $\rho = \pj{\psi} $ is a pure state then write $d_\psi$ instead of $d_{\pj{\psi}}$. Using known properties of the trace norm \cite{Nielsen}
\beq \label{def:psi_dist}
d_{\psi} (U,V) = (1 - |\inpr{\psi}{U^{\dagger}V}{\psi}|^2 )^{1/2}
\eeq
I prove a general formula later for the $L_p$  norms that includes the above as a special case.  For pure states  $d_{\psi}$ satisfies the following useful relation. 

\beq \label{eq:basic_relation} 
 \frac{1}{2}\norm{(U-e^{ix}V)\psi}^2  \leq d_{\psi}^2(U,V) \leq  \norm{(U-V)\psi}^2
\eeq
where $x$ is some real number. 
To prove these relations note that for any real $y$ 
\begin{equation*} 
\norm{(U-e^{iy}V)\psi}^2 = \inpr{\psi}{(U-e^{iy}V)^{\dagger}(U-e^{iy}V)}{\psi} = 2(1-\text{Re}\inpr{\psi}{e^{iy}U^{\dagger} V}{\psi})
\end{equation*}
Here $\text{Re}(z)$ denotes the real part of the complex number $z$. Now there is a real number $x$ such that \(\inpr{\psi}{U^{\dagger} Ve^{ix}} {\psi} \) is positive. Then \(1-\text{Re}\inpr{\psi}{U^{\dagger} Ve^{ix}} {\psi} = 1-|\inpr{\psi}{U^{\dagger} V} {\psi}| \leq 1- |\inpr{\psi}{U^{\dagger} Ve^{ix}}{\psi}|^2 = d^2_{\psi}(U,V)\) and the first inequality follows. Moreover, 
\[
\begin{split}
d^2_{\psi}(U,V) & =  (1+ |\inpr{\psi}{U^{\dagger} V}{\psi}|)(1-|\inpr{\psi}{U^{\dagger} V}{\psi}|) \\
& \leq 2(1-\text{Re}(\inpr{\psi}{U^{\dagger} V}{\psi}) 
\end{split}
\]
since  $\text{Re}(\inpr{\psi}{U^{\dagger} V}{\psi}) \leq |\inpr{\psi}{U^{\dagger} V}{\psi}|\leq 1$. 

According to  Lemma \ref{lem:pseudo1}, $d_\rho$ is a pseudometric and the following definition is the first step towards constructing a metric. Start with an arbitrary metric $\Delta$ on $\mathcal{D}$. Recall that $B(\mathcal{H})$ is the algebra of linear operators on a Hilbert space $\mathcal{H}$. 
\begin{Def}
For any two completely positive  maps (superoperators) $E, F$ on $B(\mathcal{H})$ define
\beq \label{def:uDist}
\Delta (E, F) = \sup_{\rho\in \mathcal{D}} \Delta_{\rho}(E, F) =  \sup_{\rho\in \mathcal{D}} \Delta(E\cdot\rho, F\cdot\rho) 
\eeq
\end{Def}

We can replace $\sup$ (supremum) by $\max$ (maximum) since $\mathcal{D}$ is compact and the function  $\rho \rightarrow \Delta_{\rho}(E ,F)$ is continuous for fixed $E, F$. Suppose now the metric $\Delta$ is convex function on $\mathcal{D}$:
\[ 
\begin{split}
 \Delta(\lambda \rho_1 + (1 - \lambda) \rho_2, \lambda \sigma_1 + (1 - \lambda) \sigma_2)  \leq & \;
  \lambda \Delta(\rho_1 , \sigma_1)  + ( 1 - \lambda) \Delta ( \rho_2, \sigma_2), \\
& 0 \leq \lambda \leq 1. \\
\end{split}
\]
Then the $\sup$ in \eqref{def:uDist} occurs at an extreme point of $\mathcal{D}$, that is, a pure state. This follows easily from the fact that any state $\rho$ is a convex combination of pure states and convexity of $\Delta$. In particular, if the metric is induced by a norm on $B(\mathcal{H})$ then it is convex. Hence if $\Delta$ is the trace distance ($L_1$ metric) or the Frobenius distance ($L_2$ metric) the maximum of $d_\rho(E, F)$ occurs at a pure state.

We assume  that $\Delta = \dmn{d}{1}$, the $L_1$ metric induced by the trace or $L_1$ norm. This fixing will remain true in rest of the paper except when the equivalence with metrics induced by other norms is demonstrated. To simplify notation we use $d_\rho(U,V) = \dmn{d}{1}(\!\cdot\!U\rho, V\acti\rho)$ (instead of $d_\rho^{(1)}$) and $d(U, V) = \sup_{\rho} d_\rho(U,V)$. First note that $d(U,V) = d(I, \hconj{U}V)$. Hence it suffices to study the properties of the function $d(I,W)$ for a unitary
operator $W$. Further, we may restrict  $d_\rho(1, W)$ (as a function of states) to pure states since the maximum will occur in a such a state. Since $W$ is unitary its eigenvalues lie on the unit
circle. If $z_i =e^{ic_i}$ is an eigenvalues of $W$,  then $0\leq c_i
< 2\pi$ and the angles are read counterclockwise on the unit
circle. Writing an arbitrary vector in an orthonormal  basis of eigenvectors
of $W$ we see that the numerical range\footnote{Recall that the numerical range of an operator $A$ is the set of numbers $F_A = \{ \inpr{\psi}{A}{\psi} :\: \norm{\psi} = 1\} $. } of $W$ is the {\em
  convex} set
\[F_W = \{\sum_i |x_i|^2 e^{i\theta_i} :\: \sum_i |x_i|^2 =1\}\]
That is, $F_W$ is a convex 2-polytope or polygon whose vertices lie
on the unit circle. The  number 
\[ \delta^2( 0, F_W) = \min_{\ket{\psi}} \{ | \inpr{\psi}{A}{\psi}|^2 |\; \norm{\psi} = 1\} \]
is the square of the distance from the origin to the set $F_W$  in the complex plane.  The corresponding distance of $W$ from $I$, the unit matrix is $d(I, W) = \sqrt{ 1- \delta^2( 0, F_W)}$. In the figure below the 5 eigenvalues of $W$ correspond to the points \textsf{A, B, C, D, E, G} on the unit circle. The polygon enclosed by these points is the set $F_W$. Then $\delta( 0, F_W)$ is $|\textsf{OD}|$ and $D(1, W) = |AD| = \sin{(\frac{\theta_6 - \theta_1}{2})}$. 

\begin{figure}[h]
\centering
\includegraphics[scale= .5]{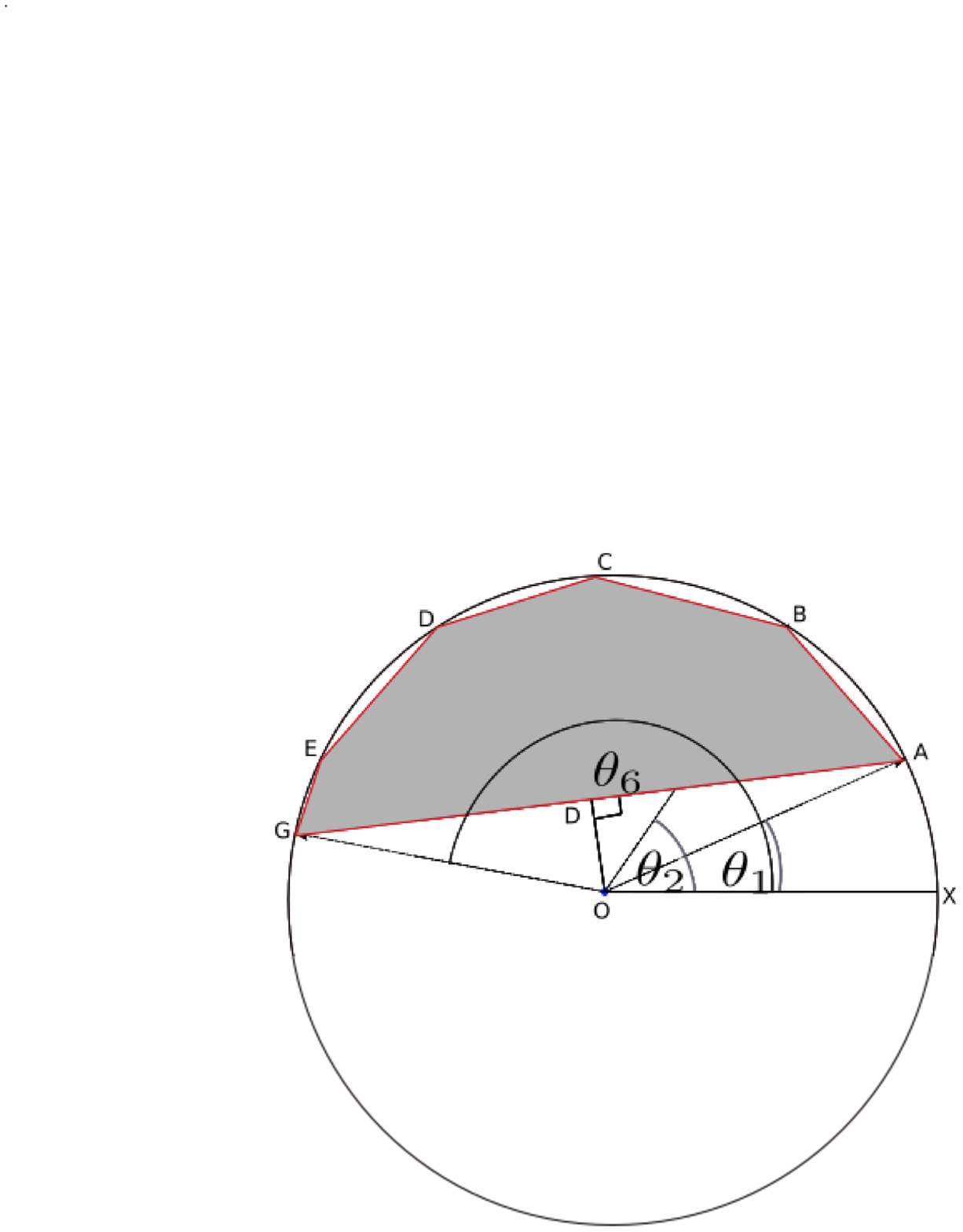}
\captionsetup{labelformat=empty}
\caption{Fig. 1}
\end{figure}
The above formula for $d(I, W)$ is intuitively obvious from the figure. In general we have the following theorem whose proof is given in the appendix. 
\begin{thm} \label{thm:eigenChar}
Let $z_i = e^{i\theta_i}, \; i= 1,\ldots n$ be the eigenvalues (possibly
with repititions) of a unitary operator $W$. Let $\theta_i$'s be ordered
such that $0\leq \theta_1\leq \theta_2\leq \cdots \leq \theta_n < 2\pi$. Let $C$ be smallest arc containing all the points $z_i$. There are two alternatives. Let $\alpha = l(C)$, the length of $C$. It is the angle subtended at the center. 
\be
\item
$C$ lies in the interior of a semicircle and hence $\alpha < \pi$. Then
\beq \label{eq:basic_formula}
d(I, W) = \sin{(\alpha/2)}
\eeq
\item
$C$ contains a semicircle: $\alpha \geq \pi$. Then $d(I, W) = 1$. 
\ee
\end{thm}
$d$ is {\em not} a metric on $\mathcal{U}_n$ because it is not necessarily positive for distinct elements. For example, $d(U, cU) = 0, |c|= 1$. However, we will continue to call it metric on $\mathcal{U}_n$ with the understanding that we identify $U$, $cU$. It is a metric on P$\mathcal{U}_n = \mathcal{U}_n/Z(\mathcal{U}_n)$, the quotient group obtained by factoring out the center. Since P$\mathcal{U}_2 \simeq \text{SO}(3)$ the metric induces a metric on the special orthogonal group in 3 dimensions. Some useful properties of the metric $d$ on $\mathcal{U}_n$ are summarized below. 
\begin{thm}\label{thm:properties}
For any pair of unitary matrices $U,V$, the function $d(U,V)$ is a pseudometric satisfying the following. 
\be[{\rm 1.}]
\item {{\rm Projective invariance.}} 
 For any complex number $c$ of modulus 1,
\[ d(U,cV) = d(cU,V) = d(U,V) \]
\item
$d(U, V) = 0$ if and only if $U = cV$ for some complex number $c$, $|c| = 1$. 
\item {{\rm Unitary invariance. }}
$d$ is invariant under left and right translations in the group
  $\mathcal{U}_n$.
\item
$d(U, V) = 1$ iff there is a unit vector $\alpha$ such that $U\alpha$
  and $V\alpha$ are orthogonal. 
\item
For unitary operators $U, V, W, X$
\beq\label{eq:monotone}
d(UV, WX) \leq d(U, W) + d(V, X)
\eeq
\ee
\end{thm}

\begin{proof}
The first two assertion follow easily  from the definition of of $d$. The fact that $d$ is left-invariant is obvious from the relation $d(U, V) = d(1, U^\dagger V) $. Next recall that $d(U, V) = [ 1- \min_{\norm{\psi}= 1} |\inpr{\psi}{U^\dagger V} {\psi}|^2]^{1/2}$. So it suffices to prove that $\min_{\norm{\psi}= 1} |\inpr{\psi}{U^\dagger V} {\psi}|$ is invariant under right translation $U\rightarrow UX, V\rightarrow VX$. Put $\ket{\phi} = X\ket{\psi}$. Then  
\[ 
 \min_{\norm{\psi } = 1} |\inpr{\psi}{X^\dagger U^\dagger VX} {\psi}| =  \min_{\norm{X\psi} = 1} |\inpr{\psi}{X^\dagger U^\dagger VX} {\psi}| 
 =  \min_{\norm{\phi} = 1} |\inpr{\phi }{U^\dagger V} {\phi}| 
\]
We use triangle inequality and translation invariance to prove 5. 
\[
d(UW, VX) \leq d(UW, VW) + d( VW, VX) \leq d(U, V) + d( W, X) 
\]
\end{proof}
Theorem \ref{thm:eigenChar} has some simple but useful consequences for tensor product of operators. Making use of the fact that the eigenvalues of the operator $A\otimes B$ are product of eigenvalues of $A$ and $B$ and the theorem we infer the following. 
\begin{propn}
Let $U,V \in \mathcal{U}_m$ and $W, X \in  \mathcal{U}_n$. Let $d(U, V) = \delta_1$ and $ d(W, X) = \delta_2$. Let $\delta = d( U\otimes W, V\otimes X)$. Then 
\[
\delta =  
\begin{cases}
  \delta_1\sqrt{1 - \delta_2^2} + \delta_2 \sqrt{ 1 - \delta_1^2} \text{ \sf{if} } \delta_1^2 + \delta_2^2 < 1 \\
 1 \text{ \sf{otherwise.}}
\end{cases}
\]
\end{propn}
\begin{proof}
Let ${U^\dagger V} = Y$ and $ W^\dagger X = Z$.  If say, $\delta_1 = 1$ then there is a state $\ket{\phi}\in H_n$ such that $\inpr{\phi}{Y}{\phi} = 0$. Then $\bra{\psi}\bra{\phi}{Y \otimes Z}\ket{\phi}\ket{\psi} = 0$ for any any $\ket{\psi} \in H_n$ proving that $\delta = 1$. In case $\delta_1, \delta_2 < 1$ as before we may assume that both the unitary operators  $Y$ and $Z$ have eigenvalues in the upper half plane, 0 is the smallest argument when ordered and $\alpha_1$ and $\alpha_2$ the largest arguments. So $\delta_1 = \sin{(\alpha_1/2)}$ and  $\delta_2 = \sin{(\alpha_2/2)}$. The largest argument of the eigenvalues of $Y\otimes Z$ is $\alpha_1 + \alpha_2$. The assertion in the proposition is a direct consequence of Theorem \ref{thm:eigenChar}. 
\end{proof}

\noindent
I conclude this section with a discussion of metrics induced by the class of $p\text{-norms},\; p \geq 1$ also called Schatten norms on operators. Given an operator $A$ on a finite dimensional Hilbert space define the $p$-norm (also called $L_p$ norm) \cite{Bhatia} by 
 \[ \norm{A}_p = [\tr(|A|^p)]^{1/p} \]
 The trace norm is the $L_1$ norm. It is easy to show that for pure states $\pj{\psi}$ and $\pj{\phi}$
 \[ \norm{ \pj{\psi} - \pj{\phi}}_p = 2^{1/p} (1 -\modulus{\inp{\psi}{\phi}}^2)^{1/2} \]
 To prove this observe that we are essentially in a 2-dimensional Hilbert space spanned by $\ket{\psi}$ and $\ket{\phi}$ and compute the trace of $|\pj{\psi} - \pj{\phi}|^p$ explicitly.  Thus the metric $\dmn{d}{p}$\footnote{Note $\dmn{d}{1} = d$ in our notation.} induced by the $p$-norm on unitary operators is a multiple of that induced by the trace norm:
\[ \dmn{d}{p} (U, V) = 2^{1/p} d(U, V) \]
Finally, it is also true that the metric induced by the so-called Ky-Fan norms \cite{Bhatia} is a multiple of $d$. 
 
\subsection{Generalisation of the pseudometrics on $\unitary{n}$}
In this subsection the metrics and pseudometrics defined earlier are generalised. Some of these generalisations are important for applications in quantum information processing and quantum mechanics discussed in the following sections.  Recall that $\mathcal{D}_n$ is the convex set of density matrices in dimension $n$. Let $K$ be a closed and convex subset of $\mathcal{D}_n$ and $U, V\in \unitary{n}$. Define 
\beq \label{def:subset-metric}
d_K (U, V)  = \sup_{\rho \in K} d_\rho(U, V)
\eeq
Clearly, $d_K$ is a pseudometric on $\unitary{n}$. We have $d_K(U, V) = 0$ iff $ U\rho U^\dagger = V  \rho V^\dagger$ for all $\rho \in K$, that is, $K$ is an {\em invariant subset} of $U^\dagger V$. Further $d_K(U, V) \leq d(U, V)$. It is still true that $d_\rho(U, V)$ attains its maximum at some extreme point as $K$ is closed and convex. But the extreme points of $K$ may include {\em mixed} states and the maximum value in $K$ may occur at such a mixed state. But if $K$ is {\em face} of $\mathcal{D}_n$,  
\footnote{Recall that a face of convex set $S$ is subset $F$ such that if $x\in F$ and $x = p x_1 + (1 - p) x_2, \; 0 <p <1 $ for $x_1, x_2 \in S$ then $x_1, x_2\in F$. That is, if a point in $F$ belongs to the interior point of a line segment then the whole line segment belongs to $F$.} 
then extreme points in $K$ will be extreme points in $\mathcal{D}_n$ and hence pure states. Consequently, $d_K(U, V)$ is attained at pure states in $K$. An important instance of a face in $\mathcal{D}_n$  is the set of density matrices with support in a subspace $\mathcal{L}$ of the underlying Hilbert space $\mathcal{H}$. A second important case is when the Hilbert space $H = H_m\otimes H_n$ and $K= K_s$ is the set of separable states in $\mathcal{D}_{mn}$
\[
\begin{split}
 K_s &= \{\rho \in \mathcal{D}_{mn} |\; \rho  = p_1\rho_1^1\rho_2^1 + \rho_1^2\rho_2^2 +\dotsc + p_k \rho_1^k\rho_2^k  , \;\rho_1^i \in \mathcal{D}_{m}, \rho_2^i \in \mathcal{D}_{n}\}\\
p_i & \geq 0 \text{ and }\sum p_i = 1\\
\end{split}
\]
$K_s$ is a closed convex subset of $\mathcal{D}_{mn}$. Moreover, any $\rho\in K_s$ can be written as a convex combination of pure {\em product} states. So the pure product states are the extreme points of $K_s$, the convex hull of the set of pure product states. The following result is useful. 
\begin{thm}\label{thm:metric_product}
Let $\mathcal{H}= \mathcal{H}_m\otimes \mathcal{H}_n$. $d_{K_s}$ is a metric on $\mathcal{U}_{mn}$ modulo its center. \footnote{Recall that it means $d$ is metric on $PU_{mn}$, the projective unitary group of dimension $mn$}
\end{thm}
\begin{proof}
Recall that we identify unitary operators $U$ and $cU, \; |c| = 1$. Let $U\in \mathcal{U}_{mn}$ act on $\mathcal{H}$. Suppose $d_{K_s}(1, U) = 0$. Then for any product state $\ket{\alpha}\otimes\ket{\beta}$ we must have $U\ket{\alpha}\otimes\ket{\beta} = a\ket{\alpha}\otimes\ket{\beta}$ where the complex number $a$ of modulus 1 may depend upon the state. So the subset $R$ of product states are partitioned by action of $U$. Thus $R_a \equiv \{\psi \in R | \; U\psi = a\psi\}$. Suppose there are two such components $R_a, R_b$. First, if $\ket{\alpha}\otimes\ket{\beta} \in R_a$ then for any $\ket{\alpha'}$, $\ket{\alpha'}\otimes \ket{\beta} \in R_a$. Suppose $\ket{\alpha'}\otimes \ket{\beta} \in R_b$. Then $U(\ket{\alpha'} + \ket{\alpha'})\otimes \ket{\beta} = (a\ket{\alpha} + b\ket{\alpha'})\otimes\ket{\beta}$. This cannot be an eigenstate of $U$ unless $a=b$. Similarly we can show that for any pair of vectors $\ket{\beta}, \ket{\beta'}$, $\ket{\alpha}\otimes\ket{\beta}$ and $\ket{\alpha}\otimes\ket{\beta'}$ belong to the same component. Computing the effect of  $U$ on $(\ket{\alpha} + \ket{\alpha'})\otimes(\ket{\beta} + \ket{\beta'})$ we note that $a=b$. That is,  $U$ acts as a constant operator $aI$ on $R$ and hence on $\mathcal{H}$. 
\end{proof}
\section{Applications of metrics and pseudometrics} 
In this section I  discuss application of metrics and pseudometrics on unitary groups in quantum mechanics, in particular, quantum information processing. 
The first application is a characterisation of distinguishable unitary operators while second  deals with use of pseudometrics in quantum coding theory. Finally,  an application of metrics to approximating unitary operators ({\em ergo} quantum circuits) is presented. Some notion of distance is needed to test the ``goodness'' of approximation. The metric $d$ is used to demonstrate quantum search as an approximation problem. 

\subsection{Distinguishing unitary operators}
Consider the following problem: {\em given that a unitary operator is picked at random from the set $\{U, V\} \subset \mathcal{U}_n$ find necessary and sufficient conditions that the chosen operator can be identified with certainty in one quantum operation}.  

We will take for granted the fact that two quantum states can be distinguished if and only if they are orthogonal (see \cite{Fuchs} for a thorough treatment). As a consequence we prove the following. 
\begin{thm}\label{thm:dist}
Two unitary operators $U,V$ are distinguishable if and only if 
\[d(U, V) = 1 \text{ or equivalently } d(I, U^\dagger V) = 1\]
\end{thm}
\begin{proof}
First suppose $d(U, V) = 1$. Then there exists a (pure) state $\ket{\alpha}$ such that $\inpr{\alpha}{U^\dagger V}{\alpha} = 0$. We apply the unitary operator $X$, randomly picked from the set $\{U, V\}$,  to $\ket{\alpha}$ followed by $U^\dagger$. The resulting state is $U^\dagger X\ket{\alpha}$. This is followed by a a projective measurement $\{P_1, P_2 = I - P\}$ where $P_1 = \pj{\alpha}$. If $X = U$ we get outcome 1 (corresponding to $P_1$) and outcome 2 if $X= V$ (for $P_2$) with probability 1. 

In the other direction, suppose two unitary operators $U$ and $V$ with $d(U, V) < 1$ can be distinguished. Let $W = U^\dagger V$. Then 
\[ |\inpr{\psi}{W}{\psi}| = \min_{\norm{\alpha} = 1} |\inpr{\alpha}{W}{\alpha}| > 0 \]
This implies that the eigenvalues of $W$ lie in a semicircle and as before we may assume it to be the upper semicircle. It is clear that $U$ and $V$ can be distinguished with certainty if and only if $I$ and $W$ can be distinguished. Using ancillary states and operators (if necessary) one will be able to distinguish with certainty if only if there exist states $\alpha, \beta$ and $\alpha', \beta'$ and a unitary operator $Z$ such that 
\[ 
\begin{split}
\inpr{\alpha'\otimes \beta'}{W\otimes Z}{\alpha\otimes \beta} &= \inpr{\alpha'}{W}{\alpha}\inpr{\beta'}{Z}{\beta} = 0 \\
|\inpr{\alpha'\otimes \beta'}{I\otimes Z}{\alpha\otimes \beta}| &= |\inp{\alpha'}{\alpha}||\inpr{\beta'}{Z}{\beta}| = 1\\
\end{split}
\]
The second line implies  $ |\inp{\alpha'}{\alpha}| = 1$ and hence from the first line it follows that $\inpr{\alpha}{W}{\alpha} = 0$. This is a contradiction. A similar reasoning shows that if we interchange 0 and 1 in the two lines we again contradict the assumption that $\inpr{\alpha}{W}{\alpha} >  0$ for all states $\alpha$. 
\end{proof}
The first half of the proof (sufficiency) goes through even when we consider the pseudometric $d_s$ over the set of separable states $K_s$. 


\subsection{Stabilizer groups and classical simulation}
Let $X, Y, Z$ denote the Pauli matrices. I use the notation standard in quantum information theory \cite{Nielsen}. Let 
\beq\label{eq:Pauli}
\mathcal{G}_1 = \{\pm I_2, \pm iI_2, \pm X,  \pm iX,  \pm Y,  \pm iY,  \pm Z,  \pm iZ \}
\eeq
where $I_2$ is the 2-dimensional unit matrix. $\mathcal{G}_1$ is a group. Let 
\[ \mathcal{G}_n = \underbrace{\mathcal{G}_1 \otimes  \mathcal{G}_1\otimes \dotsb \otimes \mathcal{G}_1}_{n \text{ factors}} \]
All elements $g\in \mathcal{G}_n$ satisfy $g^2 = \pm I$ and the eigenvalues of $g$ are either $\{1, -1\}$ or $\{i, -i\}$. Therefore, from Theorem \ref{thm:eigenChar} it follows tha $d(I, g) = 1\text{ or } 0$ and the value 0 occurs only if $g$ is in the center $\{\pm I, \pm iI \}$ of the group. Equivalently, for any $g_1, g_2\in \mathcal{G}_n$, $d(g_1, g_2) = 1 \text{ or }0$. Let $K \subset \mathcal{G}_n$ be a subgroup. A subset $S$ of the state space $\mathcal{H}_N \; (N = 2^n)$ is called a stabilizer subspace if it is the set of all vectors fixed by all $g\in K$, that is, $g\cdot \alpha = \alpha, \; \forall g\in K \text{ and } \alpha \in S$. It is clear that $S$ is a subspace. We generalize the definition of stabilizer subspace by requiring only that each $g\in K$ act as a constant operator $c(g)I$ on $S$. Any subgroup will have 0 in its stabilizer subspace. We say that a subgroup $K$ has non-trivial stabilizer if $S$ contains a non-zero vector. Let us modify the definition of the pseudometric $d_\alpha(U, V)$ to accommodate arbitrary vectors (recall that this was defined for unit vectors only in \eqref{def:psi_dist}). Thus
\[ d_{\alpha} (U, V) = \left(1 - \frac{|\inpr{\alpha}{U^\dagger V }{\alpha}|^2}{\norm{\alpha}^2}\right)^{1/2} \]
Let us first define the {\em null space} of $d$ for an {\em abelian} subgroup $H \subset \mathcal{U}_N$:
\beq\label{eq:nulls}
\mathcal{N}_H = \{\rho \in \mathcal{D}_N : \: d_\rho(I_n, U) = 0 \text{ for all } U\in H\} 
\eeq
$\mathcal{N}_H$ is convex by subadditivity of norm (trace norm in this case). Since $H$ is a subgroup $U^\dagger V \in H$ if $U, V\in H$. So $d_\rho(U, V) = d_\rho (1, U^\dagger V) = 0$ for all $\rho \in \mathcal{N}_H$. Another characterization  is 
\[ \mathcal{N}_H = \{\rho \in \mathcal{D}_N :\: U\rho = \rho U \text{ for all } U\in H\}, \]
the state operators commuting with all $U\in H$. Write any state 
\[ \rho = \sum_{i,j} a_{ij} \pjx{\psi_i}{\psi_j} \]
where the $\psi_i$ are simultaneous eigenstates of all $U\in H$. Then $\rho$ belongs to $\mathcal{N}_H$ if and only if all off-diagonal $a_{ij}= 0$ and $\rho = \sum_i r_i \pj{\psi_i}, \; r_i = a_{ii} \geq 0$ and $\sum r_i = 1$. We summerise these observations. 
\begin{lem}
Let $H$ be an abelian subgroup of $\mathcal{U}_n$. The null space $\mathcal{N}_H$ of $H$ is a convex subset. It is the convex hull of the pure states in $\mathcal{N}_H$ which are also the subset of extreme points. Further, a state $\pj{\psi}$ belongs to $\mathcal{N}_H$ if and only if $\ket{\psi}$ is a simultaneous eigenvector of all $U \in \mathcal{N}_H$. 
\end{lem} 
Let us now get back to $\mathcal{G}_n$. Let $K$ be a subgroup of $\mathcal{G}_n$. 
\begin{lem}
A subgroup $K$ will have non-trivial stabilizer subspace only if the set
\[ N_K =\left \{ \alpha \in H_N | d_{\alpha} (g, g') = 0 \; \forall g, g' \in K\right \}\] 
contains a non-zero vector. In that case, $K$ is abelian and $N_K$ is the null space of $K$. 
\end{lem}
\begin{proof}
If $\alpha \in S$ is non-zero then it is an eigenvector of $g^\dagger g' = \pm gg'$ (every element of the Pauli group is either hermitian or antihermitian). Hence $|\inpr{\alpha}{g^\dagger g'}{\alpha}|^2 = \norm{\alpha}^2$ and $d_\alpha(g, g') = 0$. Note that $g, g'$ either commute or anticommute. Suppose they anticummute. Since $d_\alpha(1, g) = d_\alpha(1, g')= d_\alpha(1, gg')= d_\alpha(1, g'g)=0$. So $\alpha$ is a common eigenvector of $g,g', gg'$ and $g'g=-gg'$. This is impossible as all eigenvalues $c$ must satisfy $|c| = 1$. 
\end{proof}

\noindent
We will assume now that $K$, the stabilzer subgroup is abelian. The corresponding null space $N_K$ is the convex hull of projections corresponding to simultaneous eigenvectors of the operators in $K$. $N_K$ is too large to be of interest. We define a stabilizer {\em subset} $\mathcal{S}$ as a closed subset  of $N_K$ that is {\em face} of $\mathcal{D}$. $\mathcal{S}_K$  being a face is convex. Now suppose $g \in K$ has two eigenvectors $\ket{\psi_1}, \ket{\psi_2}$  such that the corresponding projectors $\pj{\psi_1}, \pj{\psi_2}$ lie in $\mathcal{S}$. Then 
\[ \frac{\pj{\psi_1}}{2}+\frac{ \pj{\psi_2}}{2} = \frac{\pj{\psi_1+\psi_2}}{4} + \frac{\pj{\psi_1-\psi_2}}{4} \]
belongs to $\mathcal{S}$ by convexity. But as $\mathcal{S}$ is also a face $(\pj{\psi_1+\psi_2})/2$ and $(\pj{\psi_1 - \psi_2})/2$ also belong to it. This is possible if and only if  $\ket{\psi_1}, \ket{\psi_2}$ belong to the same eigenvalue. We now have a geometric characterization of {\em generalized} stabilizer subspaces. 
\begin{thm}
Let $K$ be an abelian subgroup of $\mathcal{U}_N$. Define a stabilizer subset $\mathcal{S}$ of $K$ as maximal face of $\mathcal{D}_N$ such that $d_{\mathcal{S}}(I, g) = 0$ for all $g\in K$. Let 
\[ S = \{\ket{\psi} \in H_N : \: \pj{\psi} \in \mathcal{S} \} \]
Then $S$ is a vector subspace such that all $g\in K$ act as constants on $S$. Thus there is a function $g\rightarrow c(g), \: |c(g)| = 1$ on $K$ such that $g\ket{\psi} = c(g)\ket{\psi}$ for all $g\in K$ and $\ket{\psi}\in S$. 
\end{thm}
Note that the subspace $S$ determines $\mathcal{S}$ in the sense sense that the  1 dimensional projectors $\pj{\psi},\: \ket{\psi}\in S$ are the extreme points of $\mathcal{S}$ and we have a geometric characterization of stabilizer subspaces. This characterization may be extended to groups other than Pauli group. 
\subsection{ Quantum Search Algorithms} 
In this subsection I consider a problem which may be considered a generalization of the search problem \cite{Grover}. Suppose we prepare a quantum system in a state
\beq\label{eq:search_gen1}
\ket{\phi} = \sin{\alpha} \ket{\psi_1} + e^{i\theta} \cos{\alpha} \ket{\psi_2},\; \inp{\psi_1}{\psi_2} = 0
\eeq
The idea is that $\ket{\psi_1}$ represents the projection of $\ket{\phi}$ into the `search' subspace. We will assume that $0 < \alpha$ is `small' (e.g. $\alpha << \pi/4$). Consider a unitary operator $U$ that sends
 \[ \ket{\phi}\rightarrow \ket{\psi_1}\text{ and } \ket{\phi^\perp} \equiv e^{-i\theta} \cos{\alpha} \ket{\psi_1} - \sin{\alpha} \ket{\psi_2} \rightarrow \ket{\psi_2} \]
 The matrix for $U$ in $\{\ket{\psi_1},  \ket{\psi_2}\}$ basis is  given by
 \[
 U = \begin{pmatrix} \sin{\alpha} & e^{-i\theta}\cos{\alpha} \\ e^{i\theta}\cos{\alpha} & -\sin{\alpha} \end{pmatrix}
 \]
 $U$ is chosen to be hermitian for convenience. We may then formulate the general search problem as: find a minimal set of unitary operators $\{V_1, \dotsc, V_k\}$ from fixed family of quantum `gates' such that the operator $V_kV_{k-1} \dotsb V_1$ is `close' to $U$  , that is, $V_kV_{k-1} \dotsb V_1$  approximate $U$. We may use a metric or even a pseudometric to estimate the goodness of the approximation and also bounds on $k$. Let us illustrate it for the metric $d$. Consider a unitary operation in the plane spanned by $\{\ket{\psi_1},  \ket{\psi_2}\}$ of the form
 \[ V =  \begin{pmatrix} \cos{\gamma} & e^{-i\theta}\sin{\gamma} \\ -e^{i\theta}\sin{\gamma} & -\sin{\gamma} \end{pmatrix}
 \]
 Then 
 \[
 \begin{split}
 & d(U, V^k)  = d(I, U V^k) = d(I, W_k) \\
 &W_k = \begin{pmatrix} \sin{(\alpha+k\gamma)} & -e^{i\theta}\cos{(\alpha + k\gamma)} \\ e^{-i\theta}\cos{(\alpha+k\gamma)} & -\sin{(\alpha+k\gamma)} \end{pmatrix}\\
 \end{split}
 \]
 Since the eigenvalues of $W_k$ are $\exp(\pm i[\pi/2 - (\alpha+k\gamma)])$ using Theorem \ref{thm:eigenChar} we conclude that $d(U, V^k) = \cos{(\alpha+k\gamma)}$. Hence the approximation may be considered good if $k\gamma \approx \pi/2 - \alpha \approx \pi/2$ (recall $\alpha$ is `small'). Using the definition of the metric $d$ this would imply that 
\[ 
\norm{\inpr{\phi}{UV^k}{\phi}} =\norm{\inpr{\psi_1}{V^k}{\phi}}  \geq \sin{(k\gamma +\alpha)}. 
\] 
Since $\sin{(k\gamma +\alpha)} \sim 1$ the transition probability $\norm{\inpr{\psi_1}{V^k}{\phi}}^2$ between the `rotated' state $V^k\ket{\phi}$ and the `search target' $\ket{\psi_1}$ is close to 1. 
If $\alpha = O(1/\sqrt{N})$ and we take $\gamma = \alpha$ then $k = O(\sqrt{N})$. Indeed, the operator $V$ with $\gamma = \alpha$ may be implemented using an {\em oracle}. 
\section{Concluding remarks}
We considered metrics and pseudometrics on `superoperators'. The main focus of investigation was (pseudo)metrics on groups of unitary operators induced by norms on the state space as these seem to have ready operational interpretations.  The three applications given in the preceding section indicate the rich potential of metrics {\em and} pseudometrics in quantum theory in general and quantum information processing in particular. Several interesting avenues for further explorations are clearly suggested. For example, one could study useful {\em pseudometrics} on the space of general superoperators. In addition, we could also investigate such structures on  product spaces and study interactions between the tensor structure and pseudometrics.  It may also prove useful to study pseudometrics on operators on infinite-dimensional spaces, a much more challenging task. 
\appendix*
\section{Proof of Theorem \ref{thm:eigenChar}}
\begin{proof}
From the definition above numerical range of $W$ is given by, 
\[ F_W = \{ \sum_i p_i z_i |\; 0 \leq p_i \leq 1 \text{ and } \sum p_i = 1 \}. \]
Moreover, we may now assume $z_i$ distinct in the convex sum. First suppose that the second condition in the theorem holds. Thus {\em every} arc containing all the $z_i$'s contains a semicircle. Let $C$ be smallest such arc.  We fix the orientation counter-clockwise. Then we can identify uniquely the endpoints of the arc as the ``starting'' and ``ending'' eigenvalues, say, $z_1= e^{i\theta_1}$ and $z_n= e^{i\theta_n}$. Replacing $W$ with $e^{i(2\pi- \theta_1)}W$ we may assume that $C$ contains the upper semicircle, $\theta_1 =0$ and $\alpha = \theta_n = \pi + \beta$ for some $0 \leq \beta < \pi$. Moreover, all other arguments satisfy $0 \leq \theta_i \leq \theta_n$ ($z_i = e^{i\theta_i}$). If, $\alpha = 0$ then the centre lies on the line joining $\theta_1$ and $\theta_n$ and hence in $F_W$ by convexity. So $d(0, F_W) = 0$ and $D(1, W) = 1$. Now assume $\alpha > 0$. There must be some eigenvalue $z_j= e^{i\theta_j}$ with $\alpha < \theta_j < \pi $ otherwise $C$ will not be smallest arc spanning all eigenvalues. Then the triangle $z_1z_jz_n$ is acute and contains the centre. By convexity again, $0 \in F_W$ and $D(1, W) = 1$. 

Next, suppose $C$ lies inside a semicircle. Since $D(1, W) = D(1, cW)$ for any complex number $c, \; |c|=1$ we may assume that $\theta_1 = 0$ and all the eigenvalues lie in the upper half plane. Then, $\alpha = \theta_n$. Again it is
intuitively clear that the line joining the points 1 and $e^{i\theta_n}$
contains the point of the polygon that is closest to the origin. To
prove it directly we have to show that
\[ \min \{|p_1+\sum_{i=2}^n p_i e^{i\theta_i} |^2 : \; 0 \leq p_i\text{ and }
\sum_i p_i=1 \} = \cos^2(\theta_n/2) \]
First, let $l$ be the line segment joining the points 1 and $z_n$. Then, it suffices
to prove that the line segment $l'$ joining the centre to an arbitrary
point of  $q$ the polygon intersects $l$ at an {\em internal}  point $z$ of
$l'$ (possibly an endpoint). For this would mean that $z$ is closer to the origin than $q$. So the point of the polygon that lies closest to the origin is the point of $l$ that is closest. The distance of this point from the origin is $\cos^2(\theta_n/2)$. Thus we have to show that, 
for any set of numbers \(\{p_1, \ldots, p_n\;|\; p_i\geq 0 \text{ and }\sum_ip_i
=1\}\) the equation
\beq \label{eq:intersect}
r\sum_i p_i e^{i\theta_i} = x e^{i\theta_1}+(1-x)e^{i\theta_n}, 
\eeq
has a unique solution with \(0\leq r \leq 1\text{ and } 0\leq x \leq 1\). As $\theta_1=0$ the above
equation is equivalent to the following pair of real equations.
\begin{flalign} \label{eq:pair} 
r(p_1+ \sum_{i=2}^n p_i\cos{\theta_i} ) & = x(1-\cos{\theta_n})+\cos{\theta_n} \\
r( \sum_{i=2}^n p_i \sin{\theta_i}) & = (1-x) \sin{\theta_n}
\end{flalign}
Hence
\beq\label{eq:intersect}
\begin{split}
x&=\frac{r(p_1+ \sum_{i=2}^n p_i\cos{\theta_i} )- \cos{\theta_n}}{(1-\cos{\theta_n})}\\
& = 1- r\frac{( \sum_{i=2}^n p_i \sin{\theta_i})} {\sin{\theta_n}} \\
\end{split}
\eeq
This implies, 
\[
\begin{split}
 & r \frac{(p_1+ \sum_{i=2}^n p_i\cos{\theta_i} )} {(1-\cos{\theta_n})}
+r\frac{( \sum_{i=2}^n p_i \sin{\theta_i})} {\sin{\theta_n}} =
  \frac{1}{1-\cos{\theta_n}}\\
& \Rightarrow  r =  \frac{\sin{\theta_n}} {(p_1\sin{\theta_n}+ \sum_{i=2}^{n-1}p_i(\sin{(\theta_n
    -\theta_i)}+ \sin{\theta_i})+p_n\sin{\theta_n}) } \\
\end{split}
\]
Let \( t = (p_1\sin{\theta_n}+ \sum_{i=2}^{n-1}p_i(\sin{(\theta_n -\theta_i)}+ \sin{\theta_i})+p_n\sin{\theta_n}) \). 
Since \(0=\theta_1\leq \theta_2 \leq \cdots \leq \theta_n <\pi \)  all $\sin{\theta_i}, \sin{(\theta_n - \theta_i)} \geq 0$. It follows that 
\[ 
\begin{split}
\sin {\theta_n} & = \sin{ (\theta_n - \theta_i) }\cos {\theta_i} +  \cos{ (\theta_n - \theta_i)} \sin {\theta_i} \\
& \leq  \sin{ (\theta_n - \theta_i) } + \sin {\theta_i} \\
\end{split}
\]
So $ t \geq p_1 \sin{\theta_n} + \sum_{i = 2}^n p_i \sin{\theta_n} = \sin{\theta_n} > 0 $. Thus $ 0 \leq r = \frac{\sin{\theta_n}}{t} \leq 1$. Substituting this value of $r$ in \eqref{eq:intersect} it is clear that $0 \leq x \leq 1$ and so $z$ is an internal point of $l'$ (and $l$). 
The theorem is proved. 
\end{proof}

\end{document}